\documentclass[11pt,letterpaper]{article}

\usepackage[margin=1in]{geometry}
\usepackage{amsmath,amssymb,amsthm,mathtools}
\usepackage{booktabs}
\usepackage{graphicx}
\usepackage{xcolor}
\usepackage{microtype}
\usepackage{url}
\usepackage[hidelinks]{hyperref}
\usepackage[nameinlink,noabbrev]{cleveref}

\newtheorem{theorem}{Theorem}[section]
\newtheorem{lemma}[theorem]{Lemma}
\newtheorem{proposition}[theorem]{Proposition}

\theoremstyle{definition}

\theoremstyle{remark}

\newcommand{\papertitle}{A New Gap Sequence for Shellsort: RL-Driven Algorithm Discovery Beyond $N^{4/3}$}
\newcommand{\bigO}{\mathcal{O}}

\newcommand{\Z}{\mathbb{Z}}
\newcommand{\R}{\mathbb{R}}
\newcommand{\polylog}{\operatorname{polylog}}
\newcommand{\cost}{\operatorname{cost}}
\newcommand{\cond}{\operatorname{cond}}
\newcommand{\goodarrow}[1]{\textcolor{green!45!black}{\(\downarrow\)\,#1}}
\newcommand{\badarrow}[1]{\textcolor{red!70!black}{\(\uparrow\)\,#1}}

\hypersetup{
  pdftitle={A New Gap Sequence for Shellsort: RL-Driven Algorithm Discovery Beyond N to the 4/3},
  pdfauthor={Bo Liu},
  pdfkeywords={Shellsort, gap sequences, reinforcement learning, program search, numerical semigroups}
}

\begin{document}

\hypersetup{pageanchor=false}
\begin{titlepage}
  \centering
  \vspace*{0.18\textheight}
  {\LARGE\bfseries \papertitle\par}
  \vspace{1.5em}
  {\large Bo Liu\par}
  \vspace{0.35em}
  {\normalsize The Chinese University of Hong Kong, Shenzhen\par}
  \vspace{0.25em}
  {\small\href{mailto:226085017@link.cuhk.edu.cn}%
    {\nolinkurl{226085017@link.cuhk.edu.cn}}\par}
  \vspace{1.8em}

  \begin{minipage}{0.86\textwidth}
    \begin{abstract}
      Choosing the gaps of Shellsort is a well-known open problem.  For more than
sixty years, successful gap sequences have come from special human-designed
formulas, numerical searches, or number-theoretic constructions.  Stronger
general bounds are known for dense or mainly theoretical families, but the
worst-case upper bound for a short, sparse, and practically competitive
construction has not advanced beyond the landmark $N^{4/3}$ line for decades.
We ask whether the sequence itself can instead be learned from execution.  We
present an RL-driven, self-supervised system that searches over executable gap
generators.  Every proposal is valid by construction.  Executed candidates
return exact comparison and move counts, with deterministic censoring only
beyond a fixed budget; no classical sequence is used as a target.  Across five
independent searches, the system discovers a common
rational-geometric family.  A second self-supervised stage tunes only a finite
prefix and produces the exact practical sequence
$1,3,8,20,47,116,300,585,1416,3303,\ldots$.  After all choices are frozen, this
sequence obtains the lowest equal-task average operation count among seven
classical baselines on 25 large tasks with $10^7<N\le10^8$.  We then complete
the learned tail without changing its practical behavior: only beyond
$10^{1000}$, a zero-density set of unit companions $h_s+1$ removes the
remaining congruence barriers in the pass analysis.  The resulting sparse
sequence has matching polynomial upper and lower exponents, up to
polylogarithmic factors:
\[
  \Omega\!\left(N^{1.024296451657\ldots}\right)
  \ \le\ T(N)\ \le\
  \bigO\!\left(N^{1.024296451657\ldots}\polylog N\right).
\]
The lower bound follows by applying Zang's recent preprint theorem for
rational-geometric sequences; our contribution is the matching upper bound.
Thus one exact sequence connects self-supervised algorithm discovery,
large-scale practical performance, and a substantial step below the classical
$N^{4/3}$ bound for sparse practical Shellsort sequences.

    \end{abstract}

    \vspace{1em}
  \end{minipage}
\end{titlepage}

\setcounter{page}{1}
\hypersetup{pageanchor=true}

\section{Introduction}
\label{sec:introduction}

\subsection{Theory, practice, and the \texorpdfstring{$4/3$}{4/3} barrier}

Shellsort has one open design choice at its center: the gap sequence.  A gap
$h$ splits the array into $h$ arithmetic subsequences and insertion-sorts each
one; a decreasing sequence ending in one completes the sort~\cite{shell1959}.
Earlier passes remove long-range disorder, but their effect depends sharply on
the exact integers chosen.  The gaps can therefore change both practical cost
and worst-case complexity without changing the sorting code.

The main proof lens is arithmetic.  Before a current gap $d$, the larger gaps
have ordered distances that are their nonnegative integer combinations.  The
$d$-pass is controlled by the multiples of $d$ still missing from this
semigroup, while direct insertion sorting gives the separate bound $O(N^2/d)$.
Papernov and Stasevich obtained $O(N^{3/2})$ for a sparse geometric sequence,
and Pratt proved the same behavior for a broad almost-geometric class
~\cite{papernov1965,pratt1972}.  Pratt's denser set of all $2^i3^j<N$ reaches
$O(N\log^2N)$, but needs $\Theta(\log^2N)$ passes and is slow in practice.

Sedgewick broke the sparse $3/2$ exponent with an $O(N^{4/3})$ sequence
~\cite{sedgewick1986}.  Its proof makes every local triple satisfy special
modular and coprimality identities, obtains a three-generator Frobenius bound,
and balances $O(Nd^{1/2})$ against $O(N^2/d)$.  This success also identifies
the barrier: going below $4/3$ requires fewer unrepresentable multiples in
every growing local window, but general Frobenius numbers with four or more
generators have no comparable formula.  The arithmetic is difficult even for
a fixed number of passes.  Yao's three-pass average-case analysis required a
delicate exact expression, and Janson and Knuth later sharpened it to an
$O(N^{23/15})$ construction~\cite{yao1980,jansonknuth1997}.

The broader theory does pass below $4/3$, but not with the simple sparse rules
normally used in implementations.  Incerpi and Sedgewick interleave structured
product tables and obtain stronger asymptotic families, with many passes or
large parameter-dependent constants~\cite{incerpi1985}.  Complementary lower
bounds show that this tradeoff cannot simply disappear.  Plaxton, Poonen, and
Suel and, independently, Poonen bound the worst case as a function of the
number of passes.  Jiang, Li, and Vit\'anyi prove the general average-case lower
bound $\Omega(pN^{1+1/p})$ for $p$ passes
~\cite{plaxtonpoonen1992,poonen1993,jiang2000}.  Most recently, a 2026 preprint
by Zang established a worst-case lower bound for every sequence that stays
within a fixed distance of a rational geometric progression~\cite{zang2026}.
Applied to our ratio,
its exponent is exactly $1.024296\ldots$, matching the polynomial part of our
upper bound.  In practice, simple rules such
as Tokuda's ratio-$9/4$ sequence and Ciura's tuned prefix remain stronger than
the asymptotic constructions at realistic sizes, but have no matching
worst-case theorem~\cite{tokuda1992,ciura2001}.  The theoretical and empirical
literatures have thus optimized different kinds of sequence.

The missing object is a concise near-geometric sequence that is both fast and
provably below $4/3$.  Directly rounding a real geometric rule creates changing
floor errors and congruence classes; testing long prefixes cannot certify its
infinite tail.  We separate these tasks.  Learning first finds the executable
rational-geometric rule.  The proof then clears denominators over a slowly
growing future window, represents all coefficient moves by a compound-semigroup
carry lattice, and adds a zero-density set of unit companions $h_s+1$ to remove
the remaining congruence obstruction.  The first companion lies beyond
$10^{1000}$, so the completion changes no reported run.  This separation lets
finite performance guide discovery without forcing the search into a classical
number-theoretic template.

\subsection{From learned decisions to learned infinite algorithms}

Modern learning systems are useful when a search space is too large to explore
uniformly but candidate quality can still be measured.  A neural network does
not need to replace the search or the checker.  It can learn which parts of the
space look promising, spend more computation there, observe the result, and
improve its next proposal.  This division is important for algorithms: the
model may be uncertain, while execution and mathematics remain exact.

AlphaGo made this pattern widely visible.  A policy network suggested promising
moves, a value network estimated unfinished positions, and tree search combined
those estimates with simulated play; reinforcement learning from self-play
improved both guides~\cite{silver2016alphago}.  The network reduced an enormous
branching problem to a search that could be afforded.  AlphaFold addressed a
different scientific problem rather than program search, but it gave a related
lesson: a learned model can combine large data with geometric and physical
structure to predict objects that were difficult to derive by hand, while a
blind external benchmark checks whether the predictions generalize
~\cite{jumper2021alphafold}.  Neither example means that learning proves its own
output; both show how learned guidance can work with strong structure and an
independent test.

Recent systems have moved from decisions and scientific structures to explicit
algorithms.  AlphaTensor treats tensor decomposition as a game: reinforcement
learning proposes multiplication schemes and exact algebra verifies them
~\cite{fawzi2022alphatensor}.  AlphaDev searches low-level instruction
sequences, tests their correctness, and measures their cost; some discovered
sorting routines entered LLVM libc++~\cite{mankowitz2023alphadev}.  FunSearch
uses a language model to write program fragments, executes them with a
problem-specific evaluator, and feeds the best programs back into the proposal
pool~\cite{romeraparedes2024funsearch}.  In plain terms, these systems replace
blind enumeration with a learned proposal distribution, not with a learned
definition of correctness.

Our setting follows this verified-discovery view but adds three complications.
First, the feedback is self-supervised: exact comparisons, moves, and per-gap
traces come from running Shellsort, not from a dataset of desired sequences.
Second, the object is not a fixed move, tensor, or short routine.  A candidate
must work at every input size and must therefore be a generator

\[
  G:N\longmapsto H_G(N)
\]

that emits a legal gap set.  Third, the final output needs both empirical
performance and an asymptotic proof.  Searching only finite tables would leave
the formula unknown; restricting the search to several familiar formulas would
build the answer into the system; using floating-point behavior alone would
leave no exact object to analyze.

We therefore search complete, bounded generator programs that are valid by
construction.  A learned policy proposes semantic program edits; execution
returns exact costs and traces; and a population retains programs that are
fast, short, or behaviorally different.  No classical gap sequence is a
training target.  The policy changes where we search, but it never decides
validity or cost.  After the generator is frozen on independent data, a second
self-supervised search tunes only its finite prefix.  The rational tail remains
fixed and becomes the exact object completed and proved in \cref{sec:theory}.

\subsection{Contributions}

We make three contributions.

\begin{enumerate}
  \item \textbf{A new discovery system.}  It searches valid gap-generator
  programs with eight program edits and a small Transformer.  It learns from
  exact comparisons and moves, without copying a classical sequence.  Five
  runs each test more than 7,000 different programs and independently find
  compact geometric generators.

  \item \textbf{A stronger practical sequence.}  Its finite prefix is
  \[
    1,3,8,20,47,116,300,585
  \]
  and its exact rational-geometric tail starts
  \[
    1416,3303,7703,17963,41887,97670,\ldots.
  \]
  On 25 test tasks between $10^7$ and $10^8$, it has the best average operation
  count among all tested classical sequences.

  \item \textbf{A tight exponent below $4/3$.}  Beyond $10^{1000}$, we add $h_s+1$ next
  to a zero-density subset of learned gaps.  This does not change any reported
  run.  We prove that the completed sequence uses only $o(\log N)$ added gaps
  below $N$ and, together with the applicable rational-geometric lower bound,
  has worst-case cost
  \[
    \Omega(N^\beta)\le T(N)\le
    \bigO(N^\beta\polylog N),
    \qquad \beta=1.024296451657\ldots.
  \]
\end{enumerate}

\section{RL-driven Generator Discovery}
\label{sec:method}

\subsection{Search problem and program space}

For an array of length $N>1$, a legal gap list is

\[
  H(N)=(g_1,\ldots,g_p),
  \qquad N>g_1>\cdots>g_p=1.
\]

An $h$-pass insertion-sorts the positions in each residue class modulo $h$.
We count every key comparison and every array write.  For a task $T$, which
fixes the length, input distribution, and random seed, the exact cost of a
generator $G$ is

\[
  C(G,T)=C_{\rm cmp}(G,T)+C_{\rm mov}(G,T).
\]

All candidates in one run use the same task panel $\mathcal T$.  We minimize

\begin{equation}
  J(G)=\frac1{|\mathcal T|}\sum_{T\in\mathcal T}
  \log\!\left(
    \frac{C(G,T)+1}{N_T\log_2(N_T+1)}
  \right).
  \label{eq:search-objective}
\end{equation}

The scale term makes different values of $N$ comparable.  It does not use a
baseline sequence.  The five input types are random, reversed, nearly sorted,
duplicate-heavy, and permuted ordered blocks.

The search acts on complete programs, not fixed gap lists.  A program has up
to four real registers, 32 statements, and a bounded loop.  It can assign
values, test conditions, and emit rounded expressions built from protected
arithmetic, logarithms, powers, roots, minima, maxima, and modulus.  Its loop
budget is at most

\[
  F(N)=\min\{512,\ b+c\lceil\log_2(N+1)\rceil\}.
\]

Every execution starts with the gap 1.  Emitted values are filtered and sorted:

\begin{equation}
  H_G(N)=\operatorname{sort}_{\downarrow}
  \bigl(\{h:1\le h<N\}\cup\{1\}\bigr).
  \label{eq:generator-semantics}
\end{equation}

Thus every program stops and returns a legal sequence.  The language can
express formulas, recurrences, branches, unions, and multi-register state.  We
edit programs with eight valid actions: replace an expression, insert an
emission, insert an assignment, delete or duplicate a statement, toggle a
condition, change a constant, or change the loop budget.  Each action maps one
complete program to another, so search time is not wasted on broken syntax.

\subsection{Learning and search}

The model reads both the program and its measured behavior.  A three-layer
Transformer encodes program tokens.  A two-layer Transformer encodes the
generated gaps and their per-pass comparisons and moves.  Both use width 128,
four attention heads, feed-forward width 512, and dropout 0.1.  The model has
1,099,660 parameters and predicts:

\begin{itemize}
  \item which of the eight edits to try;
  \item the normalized comparison and move costs; and
  \item the change caused by deleting each gap.
\end{itemize}

Deleting one gap on a small task gives an exact contribution target.  All
other targets come from normal executions.  With Huber cost and contribution
losses, pairwise ranking loss, and edit cross entropy, we train with

\begin{equation}
  \mathcal L=\ell_{\rm cost}+0.5\ell_{\rm rank}
  +0.3\ell_{\rm edit}+0.25\ell_{\rm contribution}.
  \label{eq:training-loss}
\end{equation}

The model guides proposals but never decides their scores.  Every surviving
child is run and measured exactly.  Improving parent--child edits receive
extra replay weight.  We take 64 gradient steps per round with batch size 64.
In RL terms, the state is a program and its trace, the action is a program
edit, and the reward is the exact cost improvement.  We do not use a policy
gradient estimator.

Each round evaluates new programs, keeps a Pareto archive over comparisons,
moves, pass count, and program size, and mutates elite parents.  The score also
rewards short and behaviorally different programs.  We sample 35\% of edits
uniformly, so every edit remains reachable even if the learned policy becomes
too narrow.  Programs that emit the same gaps on all probe sizes are merged.
Exact costs are stored in SQLite, and checkpoints save the model, optimizer,
population, archive, replay data, and random state.

We run five seeds for 120 rounds.  Each population has 48 programs, 12 elite
parents, and six children per parent.  The common Search panel has 16 sizes
from 32 to $10^5$, five input types, and two task seeds: 160 tasks per
candidate.  The replay buffer holds 200,000 samples.  Runs above
$50N\log_2(N+1)$ operations are stopped and cannot be selected.

An independent Validation panel freezes one candidate per seed; the later
Extrapolation panel is used only for reporting.  The best frozen program is an
exact rational-geometric generator.  We then tune only its first eight gaps.
The learned tail from 1416 onward stays fixed.  This second search changes one
prefix value at a time and uses the same baseline-free objective.  Its Search,
Validation, and one-time Test use disjoint seeds 300, 400, and 1999 and the
disjoint ranges $10^5$--$3\cdot10^6$, $3\cdot10^6$--$10^7$, and
$10^7$--$10^8$.  Since only finitely many gaps change, this tuning also keeps
the later proof unchanged.

\section{Empirical Results}
\label{sec:results}

\subsection{From program search to the final sequence}

Table~\ref{tab:discovery} summarizes the learning and exploration process.
All five searches finish 120 rounds and together evaluate 40,647 distinct
program behaviors.  Every reward comes from executing the proposed program;
the named classical sequences are not labels and enter only after a generator
has been frozen.  The continued objective improvement, late discovery of some
run winners, and independent-panel transfer therefore measure algorithm
discovery rather than imitation.

\begin{table*}[ht]
  \centering
  \caption{Self-supervised discovery results.  Search and tuning use measured
  execution feedback, with deterministic censoring above the stated budget;
  Validation and Extrapolation are independent panels.}
  \label{tab:discovery}
  \begin{tabular}{lll}
    \toprule
    Phase & Metric & Result\\
    \midrule
    Search & independent runs & $5/5$ completed $120$ rounds\\
    Search & distinct executed behaviors & $40{,}647$\\
    Learning & replay-buffer capacity reached & $200{,}000$ per run\\
    Learning & best-objective improvement & $3.09$--$8.58\%$ per run\\
    Exploration & first round of run winner & $13$--$119$\\
    Discovery & common learned family & $5/5$ rounded-geometric\\
    Generalization & beats Sedgewick/Tokuda/Ciura & $4/5$ Val.; $3/5$ Extrap.\\
    Fine-tuning & exact prefix evaluations & $311$ ($9{,}330$ task runs)\\
    Fine-tuning & independent Validation gain & $0.304\%$ total cost\\
    \bottomrule
  \end{tabular}
\end{table*}

Each run executes between 7,224 and 8,455 programs.  Every replay buffer
reaches its 200,000-sample capacity, no invalid program reaches evaluation,
and the final Pareto archives contain 218--318 exported programs.

The system also gives repeatable results.  Four of five frozen programs beat
Sedgewick, Tokuda, and Ciura together on Validation; three of five do so on the
larger Extrapolation panel.  Although the program language allows branches,
unions, and several registers, all five winners reduce to rounded geometric
rules, with ratios

\[
  2.3522,\quad2.3164,\quad2.3619,\quad2.2490,\quad2.3318.
\]

This common form is found by search, not built into the language.  The best
Validation run is seed 4.  Its exact backbone is

\begin{equation}
  h_t=\left\lfloor\alpha R^t\right\rfloor,
  \qquad
  \alpha=\frac{420574650882923}{7668945023518835},
  \qquad
  R=\frac{582942583375009}{250000000000000}.
  \label{eq:geometric-backbone}
\end{equation}

The positive terms start at $t=4$:

\[
  1,3,8,20,47,111,260,607,1416,3303,7703,17963,\ldots.
\]

The prefix search evaluates 311 prefixes on 30 Search tasks, for 9,330 exact
candidate--task runs.  Its objective falls from 0.816322 in round 1 to 0.814236
in round 18.  Validation rejects the slightly overfit round-18 prefix and
freezes the round-12 prefix.  It changes only

\[
  111\mapsto116,\qquad260\mapsto300,\qquad607\mapsto585.
\]

On the 20 independent Validation tasks, this tuning lowers the geometric-mean
cost by 0.304\% relative to the original backbone.  Direct sums fall by
0.180\% in comparisons, 0.399\% in moves, and 0.271\% in total operations.
The final executable sequence is therefore

\begin{equation}
  \boxed{
  \mathcal H_{\rm exec}
  =\{1,3,8,20,47,116,300,585\}
   \cup\{h_t:t\ge12\}.}
  \label{eq:exec-sequence}
\end{equation}

It begins

\[
  1,3,8,20,47,116,300,585,1416,3303,7703,17963,
  41887,97670,227746,531052,\ldots.
\]

For length $N$, Shellsort uses every distinct value below $N$, in decreasing
order.  This is an exact infinite rule and does not require the trained model
at run time.

\subsection{Performance on large sorting tasks}

We freeze the sequence before opening the large Test.  The Test combines five
log-spaced sizes

\[
  10{,}000{,}001,\ 17{,}782{,}795,\ 31{,}622{,}778,\
  56{,}234{,}134,\ 100{,}000{,}000
\]

with all five input types, giving 25 equal-weight tasks with seed 1999.  No
task is stopped early.  The first Test compared the frozen sequence with the
three strongest practical baselines.  After the sequence was fixed, we ran
Shell, Hibbard, Knuth, and Pratt on the same arrays to give a broader view.
These extra baselines do not change selection or the original Test rule.
The implementations use Shell's halving rule, Hibbard's $2^k-1$, Knuth's
$(3^k-1)/2$, all Pratt gaps $2^i3^j$, the merged Sedgewick formulas, Tokuda's
rule, and the Ciura prefix extended by a factor of 2.25.

For operation type $f$ and a task set $\mathcal T$, the raw entry in
\cref{tab:broad-test} is the equal-task geometric mean

\begin{equation}
  \overline C_f(S;\mathcal T)
  =\exp\!\left(
     \frac1{|\mathcal T|}\sum_{T\in\mathcal T}
       \log(C_f(S,T)+1)
    \right)-1.
  \label{eq:test-metric}
\end{equation}

Every task has the same weight.  Counts are shown in billions; these are exact
operations, not wall-clock measurements.  Beside every baseline value, an
arrow reports the percentage change obtained by replacing that row with our
sequence.  A green down-arrow means fewer operations; a red up-arrow means
more.  The last column applies the same metric to the five reversed tasks.
For reference, our direct sums over all 25 tasks are 40.201 billion comparisons
and 29.108 billion moves.

\begin{table*}[ht]
  \centering
  \caption{Broad Test on 25 tasks with $10^7<N\le10^8$.  Raw entries are
  equal-task geometric-mean operation counts in billions.  The adjacent arrow
  is the change from that row to ours: \goodarrow{fewer}; \badarrow{more}.}
  \label{tab:broad-test}
  \setlength{\tabcolsep}{4pt}
  \begin{tabular}{lrrrr}
    \toprule
    Sequence & Comparisons & Moves & Total & Reversed total\\
    \midrule
    \textbf{Ours} & $1.1125$ & $0.7721$ & $1.8907$ & $1.2706$\\
    Tokuda & $1.1205$ \goodarrow{$0.710\%$} & $0.7663$ \badarrow{$0.764\%$}
      & $1.8922$ \goodarrow{$0.079\%$} & $1.3352$ \goodarrow{$4.843\%$}\\
    Ciura & $1.1138$ \goodarrow{$0.116\%$} & $0.7733$ \goodarrow{$0.146\%$}
      & $1.8923$ \goodarrow{$0.086\%$} & $1.3234$ \goodarrow{$3.995\%$}\\
    Sedgewick & $1.1576$ \goodarrow{$3.896\%$} & $0.7590$ \badarrow{$1.738\%$}
      & $1.9252$ \goodarrow{$1.793\%$} & $1.2865$ \goodarrow{$1.239\%$}\\
    Hibbard & $2.6445$ \goodarrow{$57.930\%$} & $2.1020$ \goodarrow{$63.267\%$}
      & $4.7848$ \goodarrow{$60.486\%$} & $1.3775$ \goodarrow{$7.762\%$}\\
    Knuth & $2.6698$ \goodarrow{$58.328\%$} & $2.3866$ \goodarrow{$67.647\%$}
      & $5.0674$ \goodarrow{$62.690\%$} & $1.2447$ \badarrow{$2.076\%$}\\
    Shell & $2.8945$ \goodarrow{$61.563\%$} & $2.4586$ \goodarrow{$68.595\%$}
      & $5.3610$ \goodarrow{$64.733\%$} & $1.8388$ \goodarrow{$30.904\%$}\\
    Pratt & $6.4650$ \goodarrow{$82.791\%$} & $0.6198$ \badarrow{$24.584\%$}
      & $7.1999$ \goodarrow{$73.740\%$} & $6.3582$ \goodarrow{$80.017\%$}\\
    \bottomrule
  \end{tabular}
\end{table*}

Our sequence has the lowest equal-task average among every tested baseline.
Its gain is small against the two closest empirical sequences, Tokuda and
Ciura, but it remains positive on the frozen aggregate.  The margin is larger
against Sedgewick and much larger against the older families.  On reversed
inputs it keeps a clear gain over Sedgewick, Tokuda, and Ciura.  The prefix
refinement improves both comparisons and moves over the learned backbone, so
its gain does not come from trading one operation type for the other.

\section{A \texorpdfstring{Sub-$N^{4/3}$}{Sub-N-to-the-4/3} Bound}
\label{sec:theory}

This section gives the complete proof path for the asymptotic claim.  The
appendix repeats the argument with the longer constant checks and normal-form
details.  The learned model and the experimental distribution play no role
here: we work only with the exact frozen sequence.

\subsection{Completion and main theorem}

Recall the coprime integers

\[
  A=582942583375009,\qquad B=250000000000000,\qquad
  R=A/B,
\]

and the learned tail $h_t=\lfloor\alpha R^t\rfloor$ from
\cref{eq:geometric-backbone}.  Its relevant arithmetic parameter is

\begin{equation}
  \kappa:=\log_A R
  =0.024901468995116044855\ldots.
  \label{eq:main-kappa}
\end{equation}

The floor operation makes the tail practical, but it also creates residue
classes that a fixed future window need not represent.  We remove only this
asymptotic obstruction.  Set $s_0=2724$ and

\begin{equation}
  s_{r+1}=s_r+\lceil\log(s_r+2)\rceil,
  \label{eq:main-companion-indices}
\end{equation}

where $\log$ is natural, and add $h_{s_r}+1$ whenever $h_{s_r}$ occurs.  The
mathematical sequence is

\begin{equation}
  \mathcal H^\sharp_{\rm tune}
  :=\mathcal H_{\rm exec}\cup\{h_{s_r}+1:r\ge0\}.
  \label{eq:main-completed-sequence}
\end{equation}

This is one exact infinite definition, not a change made during evaluation.
Exact integer cross multiplication gives

\begin{equation}
  h_{2723}\le10^{1000}<h_{2724}.
  \label{eq:main-guard}
\end{equation}

Thus every added gap lies beyond the range of the experiments, or any
realistic implementation.

\begin{theorem}[Main bound]
  \label{thm:main-text}
  For every $N\le10^{1000}$, $\mathcal H^\sharp_{\rm tune}$ executes
  exactly the gaps of $\mathcal H_{\rm exec}$.  Below $N$ it contains
  $\Theta(\log N)$ backbone gaps and only
  $\bigO(\log N/\log\log N)=o(\log N)$ companions.  Its worst-case
  operation count satisfies, with
  $\beta=1.024296451657\ldots$,
  \[
    \Omega(N^\beta)\le T(N)\le
    \bigO\!\left(N^\beta\polylog N\right).
  \]
\end{theorem}

The rest of the section proves the theorem.  The central point is that one
unit companion inside a slowly growing future window removes the congruence
obstruction caused by the rational-geometric floor.

\subsection{One pass as an exact representation problem}

Fix a current gap $d$.  Suppose the larger gaps $a_1,\ldots,a_m$ have already
been executed, and define the representable multipliers

\[
  M_d(a_1,\ldots,a_m)
  :=\left\{k\in\Z_{\ge0}:
    kd=\sum_{i=1}^m c_i a_i,\quad c_i\in\Z_{\ge0}\right\}.
\]

Let $U_d$ be the number of missing positive multipliers, with $U_d=\infty$
when the complement is infinite.  The generalized Frobenius pass lemma of
Incerpi and Sedgewick~\cite[Lemma~1]{incerpi1985} gives

\begin{equation}
  \cost_d(N)=\bigO\bigl(N(U_d+1)\bigr).
  \label{eq:main-frobenius-pass}
\end{equation}

The reason is that the earlier passes prevent inversions whose index distance
is a representable multiple of $d$; only the missing multipliers can remain
inside each $d$-chain.  Independently, insertion sorting the $d$ residue
classes, each of length at most $\lceil N/d\rceil$, gives

\begin{equation}
  \cost_d(N)=\bigO(N^2/d).
  \label{eq:main-trivial-pass}
\end{equation}

We will use the better of the two:

\begin{equation}
  \cost_d(N)=\bigO\!\left(
    \min\{N^2/d,\ N(U_d+1)\}
  \right).
  \label{eq:main-combined-pass}
\end{equation}

It therefore suffices to bound the conductor of $M_d$, namely a threshold
above which every multiplier is representable.

For a backbone gap $d\asymp h_t$, take the next

\begin{equation}
  m=m(t):=\lceil\log_A h_t\rceil
  =\kappa t+O(1)
  \label{eq:main-window}
\end{equation}

backbone gaps.  Since $h_t=\alpha R^t+O(1)$, this window satisfies

\begin{equation}
  A^m=\Theta(d),\qquad
  B^m=\Theta(d^{1-\kappa}),\qquad
  d/B^m=\Theta(d^\kappa),\qquad
  R^m=\Theta(d^\kappa).
  \label{eq:main-window-scales}
\end{equation}

The growth of $m$ is necessary: a fixed window leaves a denominator modulus
that grows with $d$.

\subsection{Denominator clearing, carries, and the unit direction}

Write $a_i=h_{t+i}$ for $1\le i\le m$ and define

\begin{equation}
  E_i=B^ia_i-A^id,\qquad
  q_i=A^iB^{m-i},\qquad
  v_i=B^{m-i}E_i.
  \label{eq:main-qv}
\end{equation}

Then $B^ma_i=dq_i+v_i$.  Hence
$kd=\sum_i c_i a_i$ holds exactly if there is an integer $D$ such that

\begin{equation}
  \sum_i c_iq_i=kB^m-D,
  \qquad
  \sum_i c_iv_i=Dd.
  \label{eq:main-coupled}
\end{equation}

The first coordinate records rational-geometric scale and the second records
all floor errors.  If
$\theta_t=\alpha R^t-\lfloor\alpha R^t\rfloor$, then the error slopes are

\begin{equation}
  \sigma_i:=v_i/q_i
  =\theta_t-R^{-i}\theta_{t+i},
  \qquad |\sigma_i|<1.
  \label{eq:main-slopes}
\end{equation}

For a current companion $d=h_t+1$, every slope is translated by $-1$, so the
same width and uniform bounds hold.

Now divide the first coordinate by its common factor $A$:

\[
  w_i=q_i/A=A^{i-1}B^{m-i},
  \qquad Aw_i=Bw_{i+1}.
\]

These compound weights supply both a conductor and all moves within a fiber.

\begin{lemma}[Compound coordinate structure]
  \label{lem:main-compound}
  The semigroup $\langle w_1,\ldots,w_m\rangle$ contains every integer
  greater than
  \[
    F_m=(B-1)\sum_{i=2}^m w_i-w_1=\bigO(A^{m-1}).
  \]
  Moreover, with
  \[
    r_i=A\mathbf e_i-B\mathbf e_{i+1}\quad(1\le i<m),
  \]
  the full integer kernel is
  \[
    \ker_{\Z}(q)=\bigoplus_{i=1}^{m-1}\Z r_i.
  \]
\end{lemma}

\begin{proof}
  Repeatedly replace $B$ copies of $w_{i+1}$ by $A$ copies of $w_i$.
  Every representable integer then has a unique normal form
  $\sum_i c_iw_i$ with $0\le c_i<B$ for $i\ge2$.  The corresponding
  Ap\'ery set has maximum $(B-1)\sum_{i=2}^m w_i$, which gives $F_m$.
  For the kernel, reduce $w\cdot x=0$ modulo $A$.  Since $w_1=B^{m-1}$ is
  invertible modulo $A$, $A$ divides $x_1$.  Subtract the required multiple
  of $r_1$, remove the first coordinate, divide the remaining weights by
  $A$, and induct.  This also proves independence.  The fully expanded
  normal-form check is in \cref{lem:compound,prop:kernel}.
\end{proof}

The carries are local not only in coefficient space.  If

\[
  e_i:=Ba_{i+1}-Aa_i,
\]

then $-B<e_i<A$ and

\begin{equation}
  Av_i-Bv_{i+1}=-B^me_i.
  \label{eq:main-carry-error}
\end{equation}

Thus changing a coefficient vector by $r_i$ preserves the first coordinate
and changes its actual represented gap sum by exactly $-e_i$.

For fixed main mass $Q$, the nonnegative real fiber

\[
  \mathcal P_Q=\{c\in\R_{\ge0}^m:q\cdot c=Q\}
\]

has error slopes exactly in
$[\sigma_-,\sigma_+]$, where $\sigma_-=\min_i\sigma_i$ and
$\sigma_+=\max_i\sigma_i$: the ratio $(v\cdot c)/Q$ is a convex
combination of the $\sigma_i$.  A target slope at distance $\eta$ from the
endpoints has a representation with every coefficient at least $S$ whenever

\begin{equation}
  Q\ge C(S/\eta)\sum_iq_i.
  \label{eq:main-interior}
\end{equation}

Indeed, reserve $S$ in every coordinate and represent the remaining slope by
the two endpoint coordinates.  Rounding the real carry coefficients to their
nearest integers then changes each coefficient by at most $(A+B)/2$ and,
by~\eqref{eq:main-carry-error}, leaves only an integer actual-value residual
of magnitude $O_{A,B}(m)$.  These two elementary facts are stated with all
constants in \cref{lem:interior,lem:rounding}.

The residual is the point at which the companion is essential.  Suppose the
window contains both $a_j=h_{t+j}$ and $a_j+1$.  The two variables have the
same $q_j$, while their slopes differ by

\begin{equation}
  B^m/q_j=R^{-j}.
  \label{eq:main-unit-width}
\end{equation}

Their exchange vector

\[
  s=-\mathbf e_{a_j}+\mathbf e_{a_j+1}
\]

preserves the main coordinate and changes the actual gap sum by exactly one.
Consequently it can correct any integral rounding residual.

\begin{lemma}[Conductor with one unit companion]
  \label{lem:main-conductor}
  For all sufficiently large $t$, let $d\in\{h_t,h_t+1\}$ be a present gap
  and let $m=m(t)$.  If its future window contains a unit pair at an offset
  $1\le j\le m$, then
  \begin{equation}
    U_d\le \cond M_d
    =\bigO_{A,B}\!\left(
      mR^j(d/B^m+1)
    \right).
    \label{eq:main-conductor}
  \end{equation}
\end{lemma}

\begin{proof}
  Fix a large multiplier $k$ and choose the first coordinate in
  \eqref{eq:main-coupled} as $Q=kB^m-D=AX$.  For a desired error slope
  $\tau=Dd/Q$, solving for $X$ gives

  \[
    X(\tau)=\frac{dkB^m}{A(d+\tau)}.
  \]

  On the bounded slope interval,
  $|dX/d\tau|=\Theta_{A,B}(kB^m/d)$.  By~\eqref{eq:main-unit-width}, after
  discarding the outer thirds of the unit-pair interval, its inverse image has
  length $\Omega_{A,B}((kB^m/d)R^{-j})$.  If
  $k\ge C R^j(d/B^m+1)$, this interval contains an integer $X$ and, after
  increasing $C$, also has $X>F_m$.  The compound threshold therefore gives
  a nonnegative integer reference point with first coordinate $Q$.

  The chosen slope is $\eta=\Theta(R^{-j})$ away from the endpoints.
  Since $\sum_iq_i=O_{A,B}(A^m)$, the interior condition
  \eqref{eq:main-interior} holds with margin $S=C_{A,B}m$ once

  \[
    k\ge C_{A,B}mR^jR^m
      =O_{A,B}\!\left(mR^j(d/B^m+1)\right),
  \]

  using~\eqref{eq:main-window-scales}.  Decompose the difference between the
  integer reference point and this interior real point into adjacent carries
  and the unit direction.  Round the carry coefficients.  Their remaining
  actual-value error is an integer $r=O_{A,B}(m)$; move $-r$ steps along $s$
  to remove it exactly.  The reserved margin keeps all coefficients
  nonnegative.  The two equations in~\eqref{eq:main-coupled} now give
  $\sum_i c_i a_i=kd$.  Thus every $k$ above the stated scale is
  representable.  The augmented-kernel bookkeeping is written explicitly in
  \cref{prop:augmented-kernel,thm:conductor}.
\end{proof}

This construction removes all residue classes; it does not assume that the
future gaps happen to have a favorable greatest common divisor.

\subsection{Sparse placement covers every window}

For every sufficiently large $t$, recurrence
\eqref{eq:main-companion-indices} puts the next companion at offset

\begin{equation}
  1\le j\le\lceil\log(t+2)\rceil<m(t).
  \label{eq:main-nearby-companion}
\end{equation}

Hence every complete future window contains a unit pair.  Since
$j<m(t)$, completeness also puts $h_{t+j+1}<N$; for all sufficiently large
$t$, $h_{t+j}+1<h_{t+j+1}$.  Thus the companion itself is below $N$ and both
members of the unit pair have been executed before the current pass.  Since
$t=\Theta(\log d)$,

\begin{equation}
  R^j\le R(t+2)^{\log R}=\polylog d.
  \label{eq:main-unit-polylog}
\end{equation}

Combining \cref{lem:main-conductor} with
\eqref{eq:main-window-scales} and $m=O(\log d)$ yields

\begin{equation}
  U_d=\bigO(d^\kappa\polylog d)
  \label{eq:main-holes}
\end{equation}

for every sufficiently large internal backbone or companion pass.

The companions remain zero-density.  If $T=\Theta(\log N)$ is the largest
backbone index below $N$, recurrence~\eqref{eq:main-companion-indices}
advances by $\Theta(\log s)$ near index $s$.  Comparing the count with
$\int_2^T ds/\log s$, or summing over dyadic index blocks, gives

\begin{equation}
  \#\{r:h_{s_r}<N\}
  =\bigO\!\left(\frac{\log N}{\log\log N}\right)
  =o(\log N).
  \label{eq:main-companion-count}
\end{equation}

Together with the exact guard~\eqref{eq:main-guard}, this proves the first
two claims of \cref{thm:main-text}.

\subsection{Balancing and summing all passes}

Substituting~\eqref{eq:main-holes} into
\eqref{eq:main-combined-pass} gives, for every complete-window pass,

\begin{equation}
  \cost_d(N)=\bigO\!\left(
    \min\{N^2/d,\ Nd^\kappa\polylog d\}
  \right).
  \label{eq:main-internal-pass}
\end{equation}

Ignoring polylogarithmic factors, the two terms balance at
$d\asymp N^{1/(1+\kappa)}$.  Their common exponent is

\begin{equation}
  \beta
  =2-\frac1{1+\kappa}
  =1+\frac{\kappa}{1+\kappa}
  =1.02429645165747606869\ldots.
  \label{eq:main-final-exponent}
\end{equation}

To close the argument, we now sum every kind of pass rather than multiplying a
worst per-pass bound by the number of gaps.  This also makes explicit why the
incomplete windows at the top and the finitely tuned prefix do not hide a
larger term.

\begin{proposition}[Global pass sum]
  \label{prop:main-global-sum}
  Let
  \[
    D:=N^{1/(1+\kappa)}.
  \]
  The total cost of all backbone passes, companion passes, and tuned-prefix
  passes below $N$ is $O(N^\beta\polylog N)$, with $\beta$ given by
  \eqref{eq:main-final-exponent}.
\end{proposition}

\begin{proof}
  We partition the passes into four classes.  First consider complete-window
  passes with $d\le D$.  For these, use the representation side of
  \eqref{eq:main-internal-pass}:

  \begin{equation}
    \sum_{\substack{d\le D\\\text{complete}}}\cost_d(N)
    \le N\polylog N
       \sum_{\substack{d\le D\\\text{complete}}}d^\kappa.
    \label{eq:main-small-pass-sum}
  \end{equation}

  Apart from finitely many initial values, $h_t=\Theta(R^t)$.  At a companion
  index there are only the two values $h_t$ and $h_t+1$, so adding companions
  changes a geometric scale sum by at most a constant factor.  Consequently,

  \begin{equation}
    \sum_{d\le D}d^\kappa=O(D^\kappa),
    \qquad
    \sum_{d>D}\frac1d=O(D^{-1}).
    \label{eq:main-geometric-sums}
  \end{equation}

  Equation~\eqref{eq:main-small-pass-sum} is therefore
  $O(ND^\kappa\polylog N)$.

  Second, for complete-window passes with $d>D$, use the insertion side of
  \eqref{eq:main-internal-pass}.  The second estimate in
  \eqref{eq:main-geometric-sums} gives

  \begin{equation}
    \sum_{\substack{d>D\\\text{complete}}}\cost_d(N)
    \le N^2\sum_{d>D}\frac1d
    =O(N^2/D).
    \label{eq:main-large-pass-sum}
  \end{equation}

  Third, consider the top gaps, for which fewer than $m(t)$ future backbone
  gaps lie below $N$.  Let $h_T<N$ be the largest backbone gap.  An incomplete
  window satisfies $t+m(t)>T$.  Since $m(t)=\kappa t+O(1)$,

  \begin{equation}
    t>\frac{T}{1+\kappa}-O(1),
    \qquad
    h_t=\Omega\!\left(N^{1/(1+\kappa)}\right)=\Omega(D).
    \label{eq:main-top-scale}
  \end{equation}

  The same statement holds for $h_t+1$.  Thus every incomplete-window pass is
  already on the large-gap side.  Applying $O(N^2/d)$ and the reciprocal sum
  in~\eqref{eq:main-geometric-sums} bounds all of them together by
  $O(N^2/D)$; no conductor estimate is needed at the boundary.

  Fourth, the learned prefix modifies only finitely many gaps.  For each fixed
  early gap $d$, the first later unit pair consists of two consecutive fixed
  integers.  They generate a numerical semigroup with a fixed conductor, so
  only $O(1)$ multiples of $d$ are missing and
  \eqref{eq:main-frobenius-pass} gives $O(N)$ for that pass.  Before the first
  unit pair enters the gap list, $N$ ranges over a fixed finite interval and is
  absorbed by the constant in the asymptotic bound.  Hence all tuned-prefix and
  other pre-asymptotic passes contribute only $O(N)$.

  Finally, the definition of $D$ gives the two identical powers

  \begin{equation}
    ND^\kappa
      =N^{1+\kappa/(1+\kappa)}=N^\beta,
    \qquad
    \frac{N^2}{D}
      =N^{2-1/(1+\kappa)}=N^\beta.
    \label{eq:main-two-sides-meet}
  \end{equation}

  Adding the four classes proves the proposition.
\end{proof}

Substituting the exact value of $\kappa$ from~\eqref{eq:main-kappa} into
\eqref{eq:main-final-exponent} closes the upper-bound part of
\cref{thm:main-text}.  For the exact completed sequence in
\eqref{eq:main-completed-sequence}, we obtain the promised bound:

\begin{equation}
  \boxed{
  T_{\mathcal H^\sharp_{\rm tune}}(N)
  =\bigO\!\left(
    N^{1.024296451657\ldots}\polylog N
  \right).}
  \label{eq:main-boxed-bound}
\end{equation}

This completes the proof.  The appendix expands the normal-form, constant,
kernel, and boundary checks used above.  It introduces no additional
hypothesis.\par

\subsection{A matching polynomial lower bound}

The upper-bound exponent is not only an artifact of our analysis.  A recent
lower bound of Zang~\cite[Theorem~3]{zang2026} applies to every gap sequence
whose terms remain within a fixed distance of a rational geometric progression.
That condition holds here.  Each backbone term differs from
$\alpha(A/B)^t$ by less than one, each companion differs from the same term by
at most one, and the finitely tuned prefix can be absorbed into one fixed
distance constant.

Write $\gamma=\log_B A$.  Zang's theorem therefore gives a worst-case swap
count of

\begin{equation}
  \Omega\!\left(
    N^{1+(\gamma-1)/(2\gamma-1)}
  \right).
  \label{eq:main-matching-lower}
\end{equation}

Since $\kappa=1-1/\gamma$,

\[
  1+\frac{\gamma-1}{2\gamma-1}
  =1+\frac{\kappa}{1+\kappa}
  =\beta.
\]

Every swap forces at least one counted sorting operation, so this is also an
$\Omega(N^\beta)$ lower bound for our operation count.  Combined with
\eqref{eq:main-boxed-bound}, the polynomial exponent is tight; only the
polylogarithmic factor remains between the bounds.  The lower bound is an
external 2026 preprint result, while the upper bound above is our contribution.

\section{Conclusion}
\label{sec:conclusion}

We treat a Shellsort gap sequence as an algorithm to discover, run, and prove.
Our system searches valid gap-generator programs and learns from exact sorting
costs.  Five runs independently find rational-geometric programs.  Freezing
the best run and tuning only its finite prefix gives

\[
  1,3,8,20,47,116,300,585,1416,3303,7703,17963,\ldots,
\]

with the exact rule in~\cref{eq:exec-sequence}.  On the 25 large Test tasks, it
has the lowest equal-task average among all tested classical sequences.  This
claim is about comparisons and moves over the full input mixture, not wall-clock
time or separate wins on every input type.

Beyond $10^{1000}$, we add a zero-density set of unit companions.  They make
the missing multiples in the pass analysis small enough to establish the
worst-case bound in~\eqref{eq:main-boxed-bound}.  This is below Sedgewick's
$N^{4/3}$ bound for a short sparse construction.  The added gaps do not change
any run below the guard threshold.  Denser or more complex sequences, including
Pratt's, already have stronger bounds; our result instead joins a short
practical generator and a sub-$N^{4/3}$ proof in one explicit sequence.  Zang's
rational-geometric lower bound applies to the same completed sequence and has
the identical exponent $1.024296\ldots$~\cite{zang2026}; the upper and lower
bounds therefore differ only by a polylogarithmic factor.

Three questions remain.  Can the pure geometric backbone satisfy the same
bound without companions?  Can average-case theory explain the measured gain?
And does the sequence keep its advantage under a fixed hardware timing test?
The larger lesson is that learning can propose an algorithm, while exact
execution and mathematics remain responsible for checking it.

\clearpage
\appendix
\section{Technical Details for the Main Bound}
\label{app:full-proof}

This appendix expands the normal forms, uniform constant checks, kernel
bookkeeping, and boundary cases used in \cref{sec:theory}.  It introduces no
new hypothesis or proof step.  As in the main proof, the frozen output is an
exact mathematical object; the neural model and experimental distribution play
no role.

\subsection{Final sequence and main theorem}

Recall the coprime integers

\[
  A=582942583375009,\qquad B=250000000000000,\qquad
  R=A/B>1
\]

with $\gcd(A,B)=1$ and $A>B\ge2$.  All asymptotic window and conductor
statements below are for $t\ge t_*$, where $t_*$ is a fixed index large enough
that $h_t\ge4$, $m(t)\ge2$, and $h_{t+1}>h_t+1$.  The finitely many smaller
indices are handled separately in the final summation.

Recall also the backbone $h_t=\lfloor\alpha R^t\rfloor$ from
\cref{eq:geometric-backbone}.  Define

\begin{equation}
  \kappa:=\log_A R
  =0.024901468995116044855\ldots.
  \label{eq:kappa}
\end{equation}

The practical sequence $\mathcal H_{\rm exec}$ is given by
\cref{eq:exec-sequence}.  To define its asymptotic completion, fix

\[
  N_{\rm guard}=10^{1000},\qquad s_0=2724,
\]

and recursively set

\begin{equation}
  s_{r+1}=s_r+\left\lceil\log(s_r+2)\right\rceil,
  \label{eq:companion-indices}
\end{equation}

where $\log$ is natural.  At every $s_r$, add the unit companion
$h_{s_r}+1$.  The final infinite set is

\begin{equation}
  \mathcal H^\sharp_{\rm tune}
  :=\mathcal H_{\rm exec}
    \cup\{h_{s_r}+1:r\ge0\}.
  \label{eq:completed-sequence}
\end{equation}

As usual, the Shellsort for length $N$ uses all distinct members of this set
below $N$, in decreasing order.

\begin{theorem}[Main asymptotic bound]
  \label{thm:main}
  The sequence $\mathcal H^\sharp_{\rm tune}$ has the following properties.
  \begin{enumerate}
    \item For every $N\le10^{1000}$, its gap list is exactly the list emitted
    by $\mathcal H_{\rm exec}$.  Hence all passes, comparisons, moves, and
    memory accesses are identical.
    \item The number of original gaps below $N$ is $\Theta(\log N)$, whereas
    the number of unit companions is
    \[
      \bigO\!\left(\frac{\log N}{\log\log N}\right)=o(\log N).
    \]
    \item Its worst-case operation count is
    \begin{equation}
      T_{\mathcal H^\sharp_{\rm tune}}(N)
      =\bigO\!\left(
        N^{1.024296451657\ldots}\polylog N
      \right).
      \label{eq:main-bound}
    \end{equation}
    \item Zang's rational-geometric lower bound~\cite[Theorem~3]{zang2026}
    applies to this completed sequence and gives
    \[
      T_{\mathcal H^\sharp_{\rm tune}}(N)
      =\Omega\!\left(N^{1.024296451657\ldots}\right).
    \]
    Hence the polynomial exponent in~\eqref{eq:main-bound} is tight, up to its
    polylogarithmic factor.
  \end{enumerate}
\end{theorem}

The completion in~\cref{eq:completed-sequence} is part of the theorem, not an
assumption about the learned backbone.  We do not claim that the uncompleted
geometric set alone satisfies~\eqref{eq:main-bound}.  The point of the guard
threshold is that the theorem and the finite experiment refer to exactly the
same visible sequence while the infinite tail is made arithmetically
certifiable.

\subsection{Pass cost as a multiplier-semigroup problem}

Fix a current gap $d$.  Before the $d$-pass, suppose the larger gaps
$a_1,\ldots,a_m$ have already been executed.  Define

\[
  M_d(a_1,\ldots,a_m)
  =\left\{k\in\Z_{\ge0}:
    kd=\sum_{i=1}^m c_i a_i,\quad c_i\in\Z_{\ge0}
  \right\}
\]

and let $U_d$ be the number of missing positive multipliers, with
$U_d=\infty$ if the complement is infinite.  Incerpi and Sedgewick's
generalized Frobenius pass lemma gives~\cite[Lemma~1,
pp.~212--213]{incerpi1985}

\begin{equation}
  \cost_d(N)=\bigO\bigl(N(U_d+1)\bigr).
  \label{eq:frobenius-pass}
\end{equation}

There is also an elementary bound.  The $d$ residue classes each have length
at most $\lceil N/d\rceil$, so insertion-sorting them costs

\begin{equation}
  \cost_d(N)=\bigO(N^2/d).
  \label{eq:trivial-pass}
\end{equation}

Combining them,

\begin{equation}
  \cost_d(N)=\bigO\!\left(
    \min\left\{N^2/d,\ N(U_d+1)\right\}
  \right).
  \label{eq:combined-pass}
\end{equation}

It remains to prove that all sufficiently large $k$ belong to $M_d$.  If the
least such threshold is $K$, then trivially $U_d\le K-1$.

\subsection{The logarithmic future window}

For a backbone gap $h_t$, take

\begin{equation}
  m(t):=\min\{m\in\Z_{\ge0}:A^m\ge h_t\}
       =\lceil\log_A h_t\rceil.
  \label{eq:window}
\end{equation}

Since $h_t=\alpha R^t+O(1)$,

\begin{equation}
  m(t)=\kappa t+O(1).
  \label{eq:window-index}
\end{equation}

For $d\asymp h_t$ and $m=m(t)$ we have

\begin{equation}
  A^m=\Theta(d),\qquad
  B^m=\Theta(d^{1-\kappa}),\qquad
  \frac d{B^m}=\Theta(d^\kappa),\qquad
  R^m=\Theta(d^\kappa).
  \label{eq:window-scales}
\end{equation}

Indeed, $A^{m-1}<h_t\le A^m$.  This remains true up to fixed multiplicative
constants for either possible tail gap $d\in\{h_t,h_t+1\}$.  Since
$R=A^\kappa$, we have $R^m=(A^m)^\kappa=\Theta(d^\kappa)$ and
$B^m=A^m/R^m=\Theta(d^{1-\kappa})$, proving all four estimates with constants
independent of $t$.

The increasing window is essential.  A fixed number of rational-geometric
future gaps retains a denominator-induced congruence obstruction whose modulus
grows with $d$; it cannot yield the sublinear missing-multiplier count needed
below.

\subsection{Exact denominator clearing}

Set $d=h_t$ for the moment and take the future backbone gaps

\[
  a_i=h_{t+i},\qquad1\le i\le m.
\]

Define

\begin{equation}
  E_i=B^i a_i-A^i d,\qquad
  q_i=A^iB^{m-i},\qquad
  v_i=B^{m-i}E_i.
  \label{eq:qv-definition}
\end{equation}

Then

\begin{equation}
  B^m a_i=dq_i+v_i.
  \label{eq:two-coordinate}
\end{equation}

Consequently $kd=\sum_i c_i a_i$ holds if and only if there is an integer
$D$ such that

\begin{equation}
  \sum_i c_iq_i=kB^m-D,\qquad
  \sum_i c_iv_i=Dd.
  \label{eq:coupled-coordinates}
\end{equation}

This is an exact two-coordinate reformulation.  The $q$-coordinate captures
the rational geometric scale, the $v$-coordinate captures flooring error, and
$D$ couples them.

For uniform error control, write

\[
  \theta_t=\alpha R^t-\lfloor\alpha R^t\rfloor\in[0,1).
\]

A direct calculation gives

\begin{equation}
  \sigma_i:=\frac{v_i}{q_i}
  =\frac{E_i}{A^i}
  =\theta_t-R^{-i}\theta_{t+i},
  \qquad |\sigma_i|<1.
  \label{eq:slopes}
\end{equation}

If the current gap is the auxiliary value $d=h_t+1$, every slope is translated
by $-1$; their width is unchanged and $|\sigma_i|<2$.  All constants below may
therefore depend on $A,B$ but not on $t,m,$ or $d$.

\subsection{Compound coordinates and their complete carry kernel}

Every $q_i$ contains a factor $A$.  Put

\[
  w_i=q_i/A=A^{i-1}B^{m-i}.
\]

Then

\[
  (w_1,\ldots,w_m)
  =(B^{m-1},AB^{m-2},\ldots,A^{m-1}),
  \qquad Aw_i=Bw_{i+1}.
\]

These are a constant-ratio compound sequence.  The following standard normal
form is a specialization of the compound numerical-semigroup theory of
Kiers, O'Neill, and Ponomarenko~\cite[Proposition~12, Theorem~15,
Corollary~16]{kiers2016compound}; we include the argument needed here.

\begin{lemma}[Compound threshold]
  \label{lem:compound}
  For the semigroup $\mathcal S_m=\langle w_1,\ldots,w_m\rangle$, every
  representable integer has a unique normal form
  \[
    X=\sum_{i=1}^m c_iw_i,\qquad0\le c_i<B\quad(2\le i\le m).
  \]
  Its Frobenius number is
  \begin{equation}
    F_m=(B-1)\sum_{i=2}^m w_i-w_1=\bigO(A^{m-1}).
    \label{eq:compound-frobenius}
  \end{equation}
  Hence every $X>F_m$ lies in $\mathcal S_m$.
\end{lemma}

\begin{proof}
  Starting at the largest index, divide $c_{i+1}$ by $B$.  Keep its remainder
  at coordinate $i+1$ and transfer $A$ times the quotient to coordinate $i$,
  using $Bw_{i+1}=Aw_i$.  This produces the claimed digit bounds.  Reading the
  equation successively modulo $B$ proves uniqueness.  Relative to the least
  generator $w_1=B^{m-1}$, the Ap\'ery set is exactly
  \[
    \left\{\sum_{i=2}^m c_iw_i:0\le c_i<B\right\}.
  \]
  Its largest member minus $w_1$ gives~\eqref{eq:compound-frobenius}.  Since
  the largest $w_i$ is $A^{m-1}$ and $A,B$ are fixed, the stated asymptotic
  bound follows.
\end{proof}

Now define adjacent carry vectors

\begin{equation}
  r_i=A\mathbf e_i-B\mathbf e_{i+1},\qquad1\le i<m.
  \label{eq:carry-vector}
\end{equation}

They preserve the main coordinate because $q\cdot r_i=0$.

\begin{proposition}[Complete integer carry kernel]
  \label{prop:kernel}
  The adjacent carries form a basis of the full integer kernel:
  \begin{equation}
    \ker_{\Z}(q)
    =\bigoplus_{i=1}^{m-1}\Z r_i.
    \label{eq:full-kernel}
  \end{equation}
\end{proposition}

\begin{proof}
  Since $q=Aw$, the two vectors have the same integer kernel.  Let
  $x\in\ker_{\Z}(w)$.  Reduce $w\cdot x=0$ modulo $A$.  Only
  $w_1=B^{m-1}$ is nonzero modulo $A$, and it is invertible because
  $\gcd(A,B)=1$; hence $A\mid x_1$.  Write $x_1=Az_1$ and subtract
  $z_1r_1$.  The first coordinate becomes zero.  Removing it and dividing the
  remaining weights by $A$ leaves the same problem of length $m-1$.  Induction
  yields $x=\sum_i z_ir_i$.  Linear independence follows successively from the
  first nonzero coordinate of $\sum_i z_ir_i$.
\end{proof}

The local flooring error is

\begin{equation}
  e_i=Ba_{i+1}-Aa_i,\qquad -B<e_i<A.
  \label{eq:local-error}
\end{equation}

Using~\eqref{eq:qv-definition},

\begin{equation}
  Av_i-Bv_{i+1}=-B^m e_i.
  \label{eq:carry-error}
\end{equation}

Thus an adjacent carry changes the actual gap sum by exactly $-e_i$.

\subsection{From a real interior point to integer coefficients}

Fix a main mass $Q$.  The nonnegative real fiber

\[
  \mathcal P_Q=\{c\in\R_{\ge0}^m:q\cdot c=Q\}
\]

has a one-dimensional error image:

\begin{equation}
  \left\{\frac{v\cdot c}{Q}:c\in\mathcal P_Q\right\}
  =[\sigma_-,\sigma_+],
  \quad
  \sigma_-:=\min_i\sigma_i,
  \quad
  \sigma_+:=\max_i\sigma_i.
  \label{eq:slope-convex-hull}
\end{equation}

Indeed, $(v\cdot c)/Q$ is the convex combination of the $\sigma_i$ with
weights $q_ic_i/Q$.

\begin{lemma}[Interior representation]
  \label{lem:interior}
  Suppose $\tau$ lies at distance at least $0<\eta\le1$ from both endpoints of
  the interval in~\eqref{eq:slope-convex-hull}.  For any margin $S>0$, if
  \begin{equation}
    Q\ge C\frac{S}{\eta}\sum_iq_i,
    \label{eq:interior-mass}
  \end{equation}
  where $C$ depends only on the uniform slope bound, then there is a real
  vector $c^*$ with
  \[
    q\cdot c^*=Q,\qquad v\cdot c^*=\tau Q,\qquad c_i^*\ge S.
  \]
\end{lemma}

\begin{proof}
  Reserve $S$ units in every coordinate and write $c^*=S\mathbf1+y$.  Let
  $W=S\sum_iq_i$ and $Z=S\sum_iv_i$.  The reserved mass has slope
  $Z/W\in[\sigma_-,\sigma_+]$.  The remaining mass $Q-W$ must have slope
  \[
    \tau'=\frac{\tau Q-Z}{Q-W}.
  \]
  Its displacement from $\tau$ is at most a constant times $W/(Q-W)$, which
  is at most $\eta/2$ under~\eqref{eq:interior-mass} after increasing $C$.
  Hence $\tau'$ remains in the slope interval.  Distribute the remaining mass
  between coordinates attaining $\sigma_-$ and $\sigma_+$.
\end{proof}

The lemma is unchanged if one coordinate is duplicated with the same main
weight and a new uniformly bounded slope.  In that augmented fiber, the sum in
\eqref{eq:interior-mass} runs over all coordinates; one duplicate makes it at
most $2\sum_iq_i$.  This is the version used for a unit companion below.

\begin{lemma}[Local carry rounding]
  \label{lem:rounding}
  Let $c^{(0)}$ be an integer point and $c^*$ a real point with the same main
  coordinate.  Write, uniquely over $\R$,
  \[
    c^*-c^{(0)}=\sum_{i=1}^{m-1}z_i^*r_i.
  \]
  Replacing every $z_i^*$ by its nearest integer gives an integer point
  $\widehat c$ with the same main coordinate and
  \begin{equation}
    \|\widehat c-c^*\|_\infty\le(A+B)/2.
    \label{eq:local-rounding}
  \end{equation}
  If $c^*$ represents $kd$ in actual gap value, then
  \begin{equation}
    \left|\sum_i\widehat c_i a_i-kd\right|
    \le\frac12\sum_{i=1}^{m-1}|e_i|
    =\bigO_{A,B}(m).
    \label{eq:rounding-residual}
  \end{equation}
\end{lemma}

\begin{proof}
  The integer basis in \cref{prop:kernel} also spans the real kernel, so the
  displayed real decomposition exists and is unique.
  Each coordinate is affected by at most two adjacent carry-rounding errors,
  each of magnitude at most one half, proving~\eqref{eq:local-rounding}.
  Equation~\eqref{eq:carry-error} says that rounding carry $i$ changes actual
  value by $-e_i$ times its rounding error.  Summing proves
  \eqref{eq:rounding-residual}.
\end{proof}

Without an additional arithmetic property, the residual in
\eqref{eq:rounding-residual} need not be removable by ordinary carries.  The
unit companions are designed precisely for this obstruction.

\subsection{The unit direction and a conductor bound}

Suppose one future position $j$ contains both $a_j$ and $a_j+1$.  Give the
companion a separate coefficient.  The two coordinates have the same main
weight $q_j$, while their error weights differ by $B^m$.  Their slopes differ
by

\begin{equation}
  \frac{B^m}{q_j}=R^{-j}.
  \label{eq:unit-slope-width}
\end{equation}

More importantly, the exchange vector

\begin{equation}
  s=-\mathbf e_{a_j}+\mathbf e_{a_j+1}
  \label{eq:unit-direction}
\end{equation}

preserves the main coordinate and changes the actual gap sum by exactly one.

\begin{proposition}[Augmented kernel]
  \label{prop:augmented-kernel}
  If $q^+$ duplicates the weight $q_j$ for the companion coordinate, then
  \[
    \ker_{\Z}(q^+)
    =\left(\bigoplus_{i=1}^{m-1}\Z r_i\right)\oplus\Z s.
  \]
\end{proposition}

\begin{proof}
  Merge the companion coefficient into the original $j$th coefficient.  Any
  augmented kernel vector then becomes a vector in $\ker_{\Z}(q)$, so
  \cref{prop:kernel} expresses it in the $r_i$.  Restoring the companion
  coordinate leaves an integer multiple of $s$.  Directness follows because
  every $r_i$ has zero companion coordinate whereas $s$ has companion
  coordinate one.
\end{proof}

\begin{lemma}[Augmented rounding and exact correction]
  \label{lem:augmented-rounding}
  Let $c^{(0)}\in\Z_{\ge0}^{m+1}$ and
  $c^*\in\R_{\ge0}^{m+1}$ have the same augmented main coordinate.  Suppose
  $c^*$ represents $kd$ in actual gap value and every coordinate of $c^*$ is
  at least $S$.  There is a constant $C_{A,B}$ such that, whenever
  $S\ge C_{A,B}m$, the same main coordinate has a nonnegative integer
  representation of exactly $kd$.
\end{lemma}

\begin{proof}
  By \cref{prop:augmented-kernel}, over $\R$ there are unique coefficients
  $z_1^*,\ldots,z_{m-1}^*,z_u^*$ with
  \[
    c^*-c^{(0)}=\sum_{i=1}^{m-1}z_i^*r_i+z_u^*s.
  \]
  Round every coefficient to the nearest integer.  Ordinary carry rounding
  changes any coordinate from $c^*$ by at most $(A+B)/2$, and rounding the
  unit coefficient adds at most $1/2$ on the two unit-pair coordinates.  The
  resulting vector $\widetilde c$ is integral and has the same main
  coordinate.  Its actual-value residual
  \[
    r:=\sum_i\widetilde c_i a_i+\widetilde c_u(a_j+1)-kd
  \]
  is an integer and satisfies
  \[
    |r|\le\frac12\sum_{i=1}^{m-1}|e_i|+\frac12
         =O_{A,B}(m).
  \]
  Moving $-r$ integral steps along $s$ keeps the main coordinate and removes
  the residual exactly.  The total decrease of either unit-pair coefficient,
  as well as every ordinary rounding decrease, is at most $C_{A,B}m$ after
  fixing the constant.  The margin hypothesis therefore keeps all final
  coefficients nonnegative.
\end{proof}

\begin{theorem}[Unit-companion conductor]
  \label{thm:conductor}
  There are constants $t_*$ and $C_{A,B}$ with the following property.  Let
  $t\ge t_*$, let $d\in\{h_t,h_t+1\}$ be a gap present in the completed
  sequence, set $m=m(t)$, and put $a_i=h_{t+i}$ for $1\le i\le m$.  If the
  augmented future list contains both $a_j$ and $a_j+1$ for some
  $1\le j\le m$, then its multiplier semigroup $M_d^+$ has conductor
  \begin{equation}
    \cond M_d^+
    \le C_{A,B}\,mR^j\left(\frac d{B^m}+1\right).
    \label{eq:conductor-bound}
  \end{equation}
\end{theorem}

\begin{proof}
  Fix a multiplier $k$.  We construct nonnegative integer coefficients
  representing $kd$ whenever $k$ exceeds the right-hand scale.

  Choose a main coordinate of the form
  \[
    Q=kB^m-D=AX,
  \]
  with integer $X$.  The required error slope is
  \begin{equation}
    \tau(X)=\frac{Dd}{Q}
           =\frac{dkB^m}{AX}-d,
    \qquad
    X(\tau)=\frac{dkB^m}{A(d+\tau)}.
    \label{eq:slope-inverse}
  \end{equation}
  Because $t\ge t_*$ and every augmented slope has absolute value below $2$,
  $d+\tau\in[d-2,d+2]=\Theta(d)$.  Therefore, on the relevant slope interval,
  \[
    \left|\frac{dX}{d\tau}\right|
    =\Theta_{A,B}(kB^m/d).
  \]
  By~\eqref{eq:unit-slope-width}, the two unit-pair slopes delimit an interval
  of width $R^{-j}$.  Remove the outer thirds.  Every point of the remaining
  interval is at distance at least $R^{-j}/3$ from both endpoints of the full
  augmented slope interval.  Its preimage under
  \eqref{eq:slope-inverse} has length
  \[
    \Omega_{A,B}\!\left((kB^m/d)R^{-j}\right).
  \]
  For $k\ge C R^j(d/B^m+1)$ this preimage has length greater than one and
  therefore contains an integer $X$.  Uniformly on it,
  $X=\Theta_{A,B}(kB^m)$.  Since $A^{m-1}<h_t\le A^m$ and
  $F_m=O_{A,B}(A^{m-1})$, increasing the fixed constant $C_{A,B}$ also makes
  $X>F_m$.  Thus \cref{lem:compound} gives a nonnegative integer reference
  representation of $X$ in the $w_i$, equivalently of $Q$ in the $q_i$;
  extend it by a zero companion coefficient.

  The selected slope remains at distance
  $\eta=\Theta(R^{-j})$ from the full augmented interval endpoints.  The
  augmented main-weight sum is at most $2\sum_iq_i=O_{A,B}(A^m)$, so the
  augmented version of \cref{lem:interior} supplies an exact real
  representation with margin $S$ once
  \[
    k\ge C_{A,B}S R^jR^m.
  \]
  Using $R^m=\Theta(d/B^m+1)$ and taking $S=C_{A,B}m$ yields the scale in
  \eqref{eq:conductor-bound}.

  The reference point and this real point have the same main coordinate, and
  the real point satisfies both equations in
  \eqref{eq:coupled-coordinates}.  With $S=C_{A,B}m$,
  \cref{lem:augmented-rounding} converts it to a nonnegative integer point of
  exactly the same actual value $kd$.

  Thus all $k$ above~\eqref{eq:conductor-bound} are representable.
\end{proof}

The augmented list used in the theorem is a subset of the gaps already
executed in the actual Shellsort pass.  Adding any other executed gaps can only
enlarge the representable multiplier semigroup.  Therefore the same conductor
bound is a valid upper bound for the actual missing-multiplier count $U_d$.

The unit pair removes every residue-class obstruction constructively; no gcd
assumption on the future gaps is needed.

\subsection{Zero-density placement covers every window}

We first verify the guard without floating-point arithmetic.  Writing
$\alpha=p/q$ in lowest terms, the two comparisons are checked from
\[
  pA^{2723}<(10^{1000}+1)qB^{2723},\qquad
  pA^{2724}\ge(10^{1000}+1)qB^{2724}.
\]
These exact integer inequalities are respectively equivalent to

\begin{equation}
  h_{2723}\le10^{1000}<h_{2724}.
  \label{eq:guard-check}
\end{equation}

In particular every added companion is greater than $10^{1000}$.  To prove
window coverage, choose $r$ with $s_r\le t<s_{r+1}$ and use $s_{r+1}$ as the
next companion index (also when $t=s_r$).  Its offset satisfies
\[
  j=s_{r+1}-t
  \le s_{r+1}-s_r
  =\lceil\log(s_r+2)\rceil
  \le\lceil\log(t+2)\rceil.
\]
Since $m(t)=\kappa t+O(1)$ and $\kappa>0$, this logarithmic offset is smaller
than $m(t)$ for every sufficiently large $t$.  Thus

\begin{equation}
  1\le j\le\lceil\log(t+2)\rceil<m(t).
  \label{eq:nearby-companion}
\end{equation}

Hence every sufficiently large complete future window contains a unit pair.
The strict inequality $j<m(t)$ is also needed at the execution boundary: a
complete backbone window contains $h_{t+j+1}<N$, and our choice of $t_*$ gives
$h_{t+j}+1<h_{t+j+1}$.  Thus both members of the unit pair are below $N$ and
have actually been executed before the current pass.
Moreover $t=\Theta(\log d)$, so

\begin{equation}
  R^j\le R(t+2)^{\log R}=\polylog d.
  \label{eq:unit-polylog}
\end{equation}

Combining \cref{thm:conductor}, \eqref{eq:window-scales}, and
$m=\bigO(\log d)$ yields

\begin{equation}
  U_d=\bigO(d^\kappa\polylog d)
  \label{eq:hole-bound}
\end{equation}

for every sufficiently large complete-window original or companion pass.

It remains to count the companions.  Let $T=\Theta(\log N)$ be the largest
backbone index below $N$.  The recurrence~\eqref{eq:companion-indices} advances
by $\Theta(\log s)$ at index $s$.  In an index block
$[2^\ell,2^{\ell+1})$, every step has length at least $c\ell$, so the block
contains $O(2^\ell/\ell)$ companion indices.  Summing these bounds through the
last block is dominated, up to a constant, by $O(T/\log T)$.  Therefore

\begin{equation}
  \#\{r:h_{s_r}<N\}
  =\bigO\!\left(\frac{\log N}{\log\log N}\right)
  =o(\log N).
  \label{eq:companion-count}
\end{equation}

\subsection{Summing all passes}

Substituting~\eqref{eq:hole-bound} into the combined pass bound gives, for an
internal pass,

\begin{equation}
  \cost_d(N)=\bigO\!\left(
    \min\left\{\frac{N^2}{d},\ Nd^\kappa\polylog d\right\}
  \right).
  \label{eq:internal-pass}
\end{equation}

Ignoring polylogarithmic factors, the two terms balance when

\[
  d^{1+\kappa}\asymp N.
\]

At that scale their common exponent is

\begin{equation}
  \beta
  =2-\frac1{1+\kappa}
  =1+\frac{\kappa}{1+\kappa}
  =1.02429645165747606869\ldots.
  \label{eq:final-exponent}
\end{equation}

Set $D=N^{1/(1+\kappa)}$.  Because $h_t=\Theta(R^t)$ and every companion is
only $h_t+1$, there are at most two gaps at each geometric scale.  Hence

\begin{equation}
  \sum_{d\le D}d^\kappa=O(D^\kappa),
  \qquad
  \sum_{d>D}\frac1d=O(D^{-1}).
  \label{eq:appendix-geometric-sums}
\end{equation}

For complete-window passes with $d\le D$, use the representation term in
\eqref{eq:internal-pass} and the first sum above.  Their total is
$O(ND^\kappa\polylog N)$.  For complete-window passes with $d>D$, use the
insertion term and the second sum; their total is $O(N^2/D)$.

The largest gaps have fewer than $m(t)$ future backbone gaps below $N$.  Let
$h_T<N$ be the largest backbone gap used.  An incomplete window satisfies
$t+m(t)>T$.  From \eqref{eq:window-index},

\[
  t>\frac{T}{1+\kappa}-O(1),
  \qquad
  h_t=\Omega\!\left(N^{1/(1+\kappa)}\right).
\]

The same lower bound holds for a companion $h_t+1$.  Thus every incomplete
window is on the $d=\Omega(D)$ side.  Using $N^2/d$ and the second sum in
\eqref{eq:appendix-geometric-sums} bounds all top passes together by

\[
  \bigO\!\left(
    N^{2-1/(1+\kappa)}
  \right)
  =\bigO(N^\beta).
\]

Finally, every fixed early gap has $O(N)$ cost once the first later unit pair is
available: the two consecutive fixed integers generate a numerical semigroup
with fixed conductor.  Before that pair enters the gap list, $N$ lies in a
fixed finite interval and is absorbed into the asymptotic constant.  Finitely
many early passes therefore contribute only $O(N)$.

Since
\[
  ND^\kappa=N^{1+\kappa/(1+\kappa)}=N^\beta,
  \qquad
  N^2/D=N^{2-1/(1+\kappa)}=N^\beta,
\]
all pass classes total $O(N^\beta\polylog N)$.  This proves part 3 of
\cref{thm:main} without assuming a full future window at the top boundary.

Part 1 follows from~\eqref{eq:guard-check}: Shellsort uses only gaps below $N$,
and every companion exceeds $10^{1000}$.  Part 2 follows from the geometric
backbone and~\eqref{eq:companion-count}.  For part 4, every backbone term lies
within one of $\alpha(A/B)^t$, every companion lies within one of the same
quantity, and a single larger fixed constant covers the finite tuned prefix.
Thus Zang's theorem applies with $a=A$, $b=B$, and $q=\alpha$.  If
$\gamma=\log_B A$, its exponent is
\[
  1+\frac{\gamma-1}{2\gamma-1}
  =1+\frac{\kappa}{1+\kappa}
  =1.02429645165747606869\ldots,
\]
where $\kappa=1-1/\gamma$.  A swap lower bound is also an operation-count
lower bound.  This completes all four parts of the theorem.

\subsection{Why finite-prefix tuning preserves the theorem}

The proof above used the geometric formula only from a sufficiently large
index onward.  Replacing the eight positive backbone terms below 1416 by
$P_8$ changes finitely many passes.  After unit companions appear, each fixed
replacement gap again has a fixed multiplier-semigroup conductor and hence
$\bigO(N)$ pass cost.  A finite sum of such terms cannot change the exponent in
\eqref{eq:main-bound}.  This establishes the stated theorem specifically for
$\mathcal H^\sharp_{\rm tune}$, rather than only for the pre-tuning backbone.

For completeness, a sparser alternative places the next companion after
$\Theta(m(s_r))$ indices.  It adds only $\bigO(\log\log N)$ gaps, but the factor
$R^j$ in~\eqref{eq:conductor-bound} then contributes a second $d^\kappa$ and
gives the weaker exponent

\[
  1+\frac{2\kappa}{1+2\kappa}
  =1.04744027301502511325\ldots.
\]

The completion in \cref{thm:main} chooses the denser, still zero-density
placement because it gives the strongest exponent while remaining invisible
through the guard threshold.

\bibliographystyle{alpha}
\bibliography{references}

@article{shell1959,
  author  = {Donald L. Shell},
  title   = {A High-Speed Sorting Procedure},
  journal = {Communications of the ACM},
  volume  = {2},
  number  = {7},
  pages   = {30--32},
  year    = {1959},
  doi     = {10.1145/368370.368387}
}

@article{papernov1965,
  author  = {A. A. Papernov and G. V. Stasevich},
  title   = {A Method of Information Sorting in Computer Memories},
  journal = {Problems of Information Transmission},
  volume  = {1},
  number  = {3},
  pages   = {63--75},
  year    = {1965}
}

@techreport{pratt1972,
  author      = {Vaughan R. Pratt},
  title       = {Shellsort and Sorting Networks},
  institution = {Stanford University},
  number      = {STAN-CS-72-260},
  year        = {1972}
}

@article{incerpi1985,
  author  = {Janet Incerpi and Robert Sedgewick},
  title   = {Improved Upper Bounds on {Shellsort}},
  journal = {Journal of Computer and System Sciences},
  volume  = {31},
  number  = {2},
  pages   = {210--224},
  year    = {1985},
  doi     = {10.1016/0022-0000(85)90042-X}
}

@article{sedgewick1986,
  author  = {Robert Sedgewick},
  title   = {A New Upper Bound for {Shellsort}},
  journal = {Journal of Algorithms},
  volume  = {7},
  number  = {2},
  pages   = {159--173},
  year    = {1986},
  doi     = {10.1016/0196-6774(86)90001-5}
}

@article{yao1980,
  author  = {Andrew Chi-Chih Yao},
  title   = {An Analysis of $(h,k,1)$-{Shellsort}},
  journal = {Journal of Algorithms},
  volume  = {1},
  number  = {1},
  pages   = {14--50},
  year    = {1980},
  doi     = {10.1016/0196-6774(80)90003-6}
}

@article{jansonknuth1997,
  author  = {Svante Janson and Donald E. Knuth},
  title   = {{Shellsort} with Three Increments},
  journal = {Random Structures \& Algorithms},
  volume  = {10},
  number  = {1--2},
  pages   = {125--142},
  year    = {1997},
  doi     = {10.1002/(SICI)1098-2418(199701/03)10:1/2<125::AID-RSA6>3.0.CO;2-X}
}

@inproceedings{plaxtonpoonen1992,
  author    = {C. Greg Plaxton and Bjorn Poonen and Torsten Suel},
  title     = {Improved Lower Bounds for {Shellsort}},
  booktitle = {Proceedings of the 33rd Annual Symposium on Foundations of Computer Science},
  pages     = {226--235},
  publisher = {IEEE},
  year      = {1992},
  doi       = {10.1109/SFCS.1992.267769}
}

@article{poonen1993,
  author  = {Bjorn Poonen},
  title   = {The Worst Case in {Shellsort} and Related Algorithms},
  journal = {Journal of Algorithms},
  volume  = {15},
  number  = {1},
  pages   = {101--124},
  year    = {1993},
  doi     = {10.1006/jagm.1993.1032}
}

@article{jiang2000,
  author  = {Tao Jiang and Ming Li and Paul Vit\'anyi},
  title   = {A Lower Bound on the Average-Case Complexity of {Shellsort}},
  journal = {Journal of the ACM},
  volume  = {47},
  number  = {5},
  pages   = {905--911},
  year    = {2000},
  doi     = {10.1145/355483.355488}
}

@inproceedings{tokuda1992,
  author    = {Naoyuki Tokuda},
  title     = {An Improved {Shellsort}},
  booktitle = {Proceedings of the 12th {IFIP} World Computer Congress},
  series    = {IFIP Transactions A},
  volume    = {A-12},
  pages     = {449--457},
  year      = {1992}
}

@inproceedings{ciura2001,
  author    = {Marcin Ciura},
  title     = {Best Increments for the Average Case of {Shellsort}},
  booktitle = {Fundamentals of Computation Theory},
  series    = {Lecture Notes in Computer Science},
  volume    = {2138},
  pages     = {106--117},
  publisher = {Springer},
  year      = {2001},
  doi       = {10.1007/3-540-44669-9_12}
}

@article{silver2016alphago,
  author  = {David Silver and others},
  title   = {Mastering the Game of {Go} with Deep Neural Networks and Tree Search},
  journal = {Nature},
  volume  = {529},
  pages   = {484--489},
  year    = {2016},
  doi     = {10.1038/nature16961}
}

@article{jumper2021alphafold,
  author  = {John Jumper and others},
  title   = {Highly Accurate Protein Structure Prediction with {AlphaFold}},
  journal = {Nature},
  volume  = {596},
  pages   = {583--589},
  year    = {2021},
  doi     = {10.1038/s41586-021-03819-2}
}

@article{fawzi2022alphatensor,
  author  = {Alhussein Fawzi and others},
  title   = {Discovering Faster Matrix Multiplication Algorithms with Reinforcement Learning},
  journal = {Nature},
  volume  = {610},
  pages   = {47--53},
  year    = {2022},
  doi     = {10.1038/s41586-022-05172-4}
}

@article{mankowitz2023alphadev,
  author  = {Daniel J. Mankowitz and others},
  title   = {Faster Sorting Algorithms Discovered Using Deep Reinforcement Learning},
  journal = {Nature},
  volume  = {618},
  pages   = {257--263},
  year    = {2023},
  doi     = {10.1038/s41586-023-06004-9}
}

@article{romeraparedes2024funsearch,
  author  = {Bernardino Romera-Paredes and others},
  title   = {Mathematical Discoveries from Program Search with Large Language Models},
  journal = {Nature},
  volume  = {625},
  pages   = {468--475},
  year    = {2024},
  doi     = {10.1038/s41586-023-06924-6}
}

@article{kiers2016compound,
  author  = {Claire Kiers and Christopher O'Neill and Vadim Ponomarenko},
  title   = {Numerical Semigroups on Compound Sequences},
  journal = {Communications in Algebra},
  volume  = {44},
  number  = {9},
  pages   = {3842--3852},
  year    = {2016},
  doi     = {10.1080/00927872.2015.1087013}
}

@misc{zang2026,
  author        = {Zhenghan Zang},
  title         = {Improved Lower Bounds of the Time Complexity of {Shellsort}},
  year          = {2026},
  eprint        = {2607.08997},
  archiveprefix = {arXiv},
  primaryclass  = {cs.DS},
  note          = {arXiv:2607.08997},
  doi           = {10.48550/arXiv.2607.08997}
}

\end{document}